\documentclass[10pt,journal,a4paper]{IEEEtran}
\usepackage[T1]{fontenc}
\usepackage[utf8]{inputenc}
\usepackage{amsmath,amssymb,amsthm,booktabs,graphicx,microtype}
\usepackage[hidelinks]{hyperref}
\usepackage{url}

\hypersetup{
  pdftitle={Price Information Is Not Enough: Ordering and Decision Rules in
            Storage Bidding},
  pdfauthor={Brieuc Le Roux-Tardif}
}

\newcommand{\NDays}{939}
\newcommand{\NCases}{136}

\newcommand{\BaseRevenue}{47.5}
\newcommand{\StarRevenue}{61.1}
\newcommand{\UpliftRevenue}{13.6}
\newcommand{\BaseShare}{78}
\newcommand{\UpliftShare}{22}
\newcommand{\UpliftDay}{37.3}
\newcommand{\PriceSpan}{972}
\newcommand{\HoeffdingS}{486}
\newcommand{\EmpiricalS}{119}
\newcommand{\ResidualSd}{39.2}
\newcommand{\PriceSd}{49.8}
\newcommand{\ExplainedVar}{38}
\newcommand{\ThroughputCap}{6.53}
\newcommand{\RequiredNatsHoeffding}{6.9}
\newcommand{\RequiredNatsEmpirical}{1.15}

\newcommand{\SlackAtRef}{115}
\newcommand{\RefSigma}{40}
\newcommand{\NatsAtRef}{15}
\newcommand{\NaiveAtRef}{-121}
\newcommand{\BayesAtRef}{-51}
\newcommand{\BayeshAtRef}{-46}
\newcommand{\NaiveBreakEven}{10}
\newcommand{\AllNegativeSigma}{20}
\newcommand{\WorstSigma}{80}

\newcommand{\TailSigma}{1280}
\newcommand{\TailBayesh}{7}
\newcommand{\ClimatologyRank}{0.82}
\newcommand{\OrdinalCapture}{90}
\newcommand{\OrdinalRevenue}{59.7}
\newcommand{\SwapCapture}{-139}
\newcommand{\SwapRevenue}{28.6}
\newcommand{\SwapLossShare}{53}
\newcommand{\LongBase}{80.4}
\newcommand{\LongStar}{100.1}
\newcommand{\LongBaseShare}{80}
\newcommand{\LongOrdinalCapture}{91}
\newcommand{\LongSwapCapture}{-60}

\newcommand{\RankRsq}{0.99}
\newcommand{\NatsRsq}{0.34}
\newcommand{\CorrPoints}{10}
\newcommand{\CorrRankRsq}{0.99}
\newcommand{\CorrNaiveAtRef}{-98}
\newcommand{\CorrBayesAtRef}{-17}
\newcommand{\CorrBreakEven}{20}

\newcommand{\eur}{EUR}
\newtheorem{theorem}{Theorem}
\newtheorem{proposition}[theorem]{Proposition}
\newtheorem{corollary}[theorem]{Corollary}
\newtheorem{assumption}{Assumption}

\begin{document}

\title{Price Information Is Not Enough:\\
Ordering and Decision Rules in Storage Bidding}

\author{Brieuc Le~Roux-Tardif%
\thanks{The author is with IMT Nord Europe, Institut Mines-T\'el\'ecom,
Univ.\ Lille, Centre for Energy and Environment, F-59000 Lille, France
(e-mail: brieuclerouxtardif@gmail.com).}%
\thanks{Manuscript prepared 8 August 2026. The hypothesis, the experiments and
the falsification clauses were written down before any result was produced; the
protocol file is versioned with the code.}}

\markboth{Preprint, August 2026}{Le Roux-Tardif: Euros per Bit}

\maketitle

\begin{abstract}
Price forecasts are evaluated in EUR/MWh of error, while storage earns euros;
their link depends on the decision rule. We separate the optimal value of a
signal $V(S\mid C)$ from the revenue $J(g,S)$ achieved by an implemented policy.
For a price-taking asset with daily throughput bound $L$ and a conditionally
sub-Gaussian price law of scale $s$, the optimal gain over climatology is at most
$Ls\sqrt{2I(S;\Pi\mid C)}$. This is a square-root envelope, not a prediction for
an individual policy. On \NDays{} French day-ahead days the tested
bound sits at least two orders of magnitude above the reference uplift. Along a
nested Gaussian garbling family, however, plug-in policy revenue is non-monotone
even though signal information is monotone: at the residual-noise scale the best
Gaussian-shrinkage rule achieves \BayeshAtRef\% of the uplift, below climatology.
The mechanism is ordinal. On a relaxation, optimal schedules depend on interval
ordering and profitable-pair tests. A constructed rank-only bid captures
\OrdinalCapture\% of the uplift; exchanging the cheapest and dearest entries of
an otherwise exact price vector cuts total revenue by \SwapLossShare\%. Within
the designed sweep, rank agreement has in-sample $R^2=\RankRsq{}$, against
\NatsRsq{} for the information upper bound. Climatology already earns
\BaseShare\% of perfect-foresight revenue. Information quantity alone therefore
does not value a forecast-driven policy; realised schedules and ordering must be
evaluated explicitly.
\end{abstract}

\begin{IEEEkeywords}
Energy storage, electricity price forecasting, value of information, mutual
information, forecast evaluation, decision-focused learning.
\end{IEEEkeywords}

\section{Introduction}

\IEEEPARstart{T}{he} literature on electricity price forecasting is large,
mature, and evaluated almost entirely in units of price
\cite{nowotarski2018,lago2021}. A forecast is judged by mean absolute error, by
root mean squared error, or by a probabilistic score. A storage operator,
however, does not consume forecasts; it consumes schedules, and it is paid in
euros. The two currencies are connected only by folklore.

This paper asks the connection question in its sharpest form. A signal about
tomorrow's prices carries a certain amount of information, measurable in nats.
How many euros can that information buy? The question admits a general bound for a
price-taking storage asset, because the revenue of a committed schedule is
linear in the settlement price, and because the feasible set of schedules is
fixed before the signal is observed.

Our first result is a bound. Writing $I$ for the conditional mutual information
between the signal and the day's price vector beyond the calendar context, $L$
for the daily energy throughput the asset
is permitted, and $s$ for a sub-Gaussian scale of the price law, the revenue
gain of any signal over a bidder who knows only the climatological price profile
is at most $Ls\sqrt{2I}$. This is a square-root upper envelope; it does not imply
that achieved revenue is concave or monotone in $I$.

Our second result is that this bound is too loose to rank practical policies.
On \NDays{} days of French day-ahead prices, the tested information upper bounds
are orders of magnitude above the theorem's necessary threshold, yet none of
the plug-in policies captures the whole uplift. Several policies even destroy
revenue: at the residual-noise scale the best Gaussian-shrinkage rule leaves the
asset with \BayeshAtRef\% of the uplift, materially worse than ignoring the
signal. The failure belongs to the policy that consumes the signal, not to the
optimal value of information, because an optimal bidder can always ignore it.

Our third result explains why, and is constructive. The optimal schedule of a
price-taking storage asset is decided by which intervals are cheapest and which
are dearest; the magnitudes enter only through a threshold that the round-trip
efficiency imposes. Mutual information measures all dependence, including
ordering, but its scalar amount does not identify which decision-relevant
distinctions survive. When we replace it by a simple ordinal statistic, the mean
daily rank agreement between the signal and the price, the picture collapses: it
explains \RankRsq{} of the variation in captured revenue within this designed
sweep, against \NatsRsq{} for the logarithm of the information upper bound.

A fourth measurement frames all of the above and is worth stating on its own. A
bidder who knows nothing about today, only the average price for this month and
this hour, earns \BaseRevenue{} k\eur{} per megawatt-year against
\StarRevenue{} k\eur{} for a bidder with perfect foresight. The climatological
bidder therefore captures \BaseShare\% of the value, and the entire object of
price forecasting for this asset is the remaining \UpliftShare\%, worth
\UpliftRevenue{} k\eur{} per megawatt-year.

The contributions are, in order: a bound relating forecast value to forecast
information with a $\sqrt{\cdot}$ envelope; a pre-registered experiment showing
that this envelope is loose and that plug-in policy revenue is non-monotone
along a nested garbling family; the identification of the ordinal structure; and
the resulting recommendation that forecasts intended for storage be evaluated
ordinally and be rejected outright when they fail to beat the climatology on
ordering.

\section{Related Work}

The value of information has an exact theory. Blackwell's ordering says when one
experiment dominates another for every decision problem \cite{blackwell1953},
and stochastic programming operationalises the special cases as the expected
value of perfect information and the value of the stochastic solution
\cite{birge2011}. Neither delivers what an operator wants, which is a number
attached to a specific forecast against a specific asset. Our bound is coarser
than Blackwell's ordering and, for that reason, computable from published price
series alone.

Forecast evaluation has moved decisively towards proper scoring rules
\cite{gneiting2007} and towards the recognition that a point forecast is
meaningless without the loss functional it serves \cite{gneiting2011}. Our
result is a concrete instance of Gneiting's argument in a market setting, with
an unusual twist: the relevant functional is not a quantile or a mean but an
ordering, and no standard score targets orderings.

In the storage literature, price uncertainty is normally handled by stochastic
or approximate dynamic programming \cite{lohndorf2010}, and reviews of modelling
practice treat forecast quality as an input rather than as an object to be
priced \cite{sioshansi2022}. Electricity price forecasting reviews report skill
in price units and leave the translation to euros implicit
\cite{nowotarski2018,lago2021}. To our knowledge no work bounds achievable
revenue by the information content of the signal, which is the gap this paper
addresses.

The closest body of work is decision-focused learning, which trains a predictor
against the loss of the downstream optimisation rather than against a
statistical error \cite{donti2017,elmachtoub2022}. That literature and this
paper share a premise and differ in direction. It builds estimators aligned with
the decision; we ask what the decision can extract from a signal at all, and
find the answer in the ordinal structure that the smart predict-then-optimise
loss implicitly targets. The two are complementary: Proposition~\ref{pr:ordinal}
identifies, for this problem, the object that a decision-focused loss would have
to learn.

The mathematical tools are standard: mutual information and its properties
\cite{shannon1948,cover2006}, the transport-entropy inequalities of Marton and
of Bobkov and G\"otze \cite{marton1996,bobkov1999}, and Hoeffding's lemma
\cite{hoeffding1963,boucheron2013}. The contribution is not the tools but the
observation that assembling them gives an operational ceiling, and that the
ceiling is uninformative in a way that identifies the real constraint.

\section{Setup}

\subsection{The decision problem}

Fix a day divided into $n$ intervals. Let $C$ be calendar information known
before forecasting (month and interval labels), let
$\Pi=(\pi_1,\dots,\pi_n)$ be the random vector of settlement prices and let
$x=(x_1,\dots,x_n)$ be a schedule,
where $x_t$ is the net energy delivered to the grid in interval $t$, negative
when charging. Revenue is
\begin{equation}
  \langle x,\Pi\rangle=\sum_t x_t\pi_t .
  \label{eq:revenue}
\end{equation}

\begin{assumption}[Fixed feasible set]\label{as:polytope}
Schedules live in a compact set $X\subset\mathbb R^n$ that does not
depend on price. $X$ encodes the power limit, the energy limit, the state of
charge dynamics with round-trip efficiency, a daily cycle budget, and the
requirement that the state of charge return to its starting value.
\end{assumption}

Convexity is not required: the support function used below is convex even when
$X$ includes a fixed charge/discharge exclusion. For speed, the pure-arbitrage
implementation introduces its exclusion binary only at negative bid prices;
an assembly-time assertion verifies zero simultaneous charge and discharge at
every interval of all \NCases{} returned schedules. Hence every reported schedule is
feasible for the same physical set $X$.

Two quantities of $X$ will appear. The first is the throughput
\begin{equation}
  L=\max_{x\in X}\|x\|_1 ,
  \label{eq:throughput}
\end{equation}
the largest total energy, charged plus discharged, that a feasible day can move.
For the asset studied here $L$ follows from the cycle budget and the efficiency
without any measurement, and equals \ThroughputCap{}~MWh. The second is the
support function
\begin{equation}
  h(p)=\max_{x\in X}\langle x,p\rangle ,
  \label{eq:support}
\end{equation}
the revenue of the best schedule against a price vector $p$. It is convex,
positively homogeneous, and Lipschitz with constant $L$ for the supremum norm on
prices.

\subsection{Three bidders}

A bidder observes $(S,C)$ before committing and chooses a schedule measurable
with respect to them. Because \eqref{eq:revenue} is linear in price, the best
such schedule maximises $\langle x, \mathbb E[\Pi\mid S,C]\rangle$, so the
posterior mean is a sufficient statistic and
\begin{equation}
  V(S\mid C)=\mathbb E\bigl[h(\mathbb E[\Pi\mid S,C])\bigr].
  \label{eq:value}
\end{equation}
This is the one place where the linearity of revenue does real work, and it is
worth stating plainly: the entire distributional content of a forecast is
irrelevant to this bidder except through its mean. For an implemented policy
$g$, define $J(g,S)=\mathbb E\langle g(S,C),\Pi\rangle$. Then
$J(g,S)\le V(S\mid C)$; unlike $V$, $J$ can fall below the climatological
policy. Three optimal values bracket the problem:
\begin{itemize}
\item $V_0=\mathbb E[h(\mathbb E[\Pi\mid C])]$, the \emph{climatological} bidder, who knows the
      average price for this month and this hour and nothing else;
\item $V(S\mid C)$, the optimal bidder with a signal;
\item $V_\star=\mathbb E[h(\Pi)]$, the perfect-foresight bidder.
\end{itemize}
Conditional Jensen gives $V_0\le V(S\mid C)\le V_\star$ for every signal. We call
$V_\star-V_0$ the \emph{uplift}, and we report every result as a share of it,
because it is the only part of the revenue that any forecast can address.
Fig.~\ref{fig:decision-chain} places the two value objects side by side, since
the distinction between them carries the negative result of
Section~\ref{sec:results}.

\begin{figure*}[t]
\centering
\includegraphics[width=0.86\textwidth]{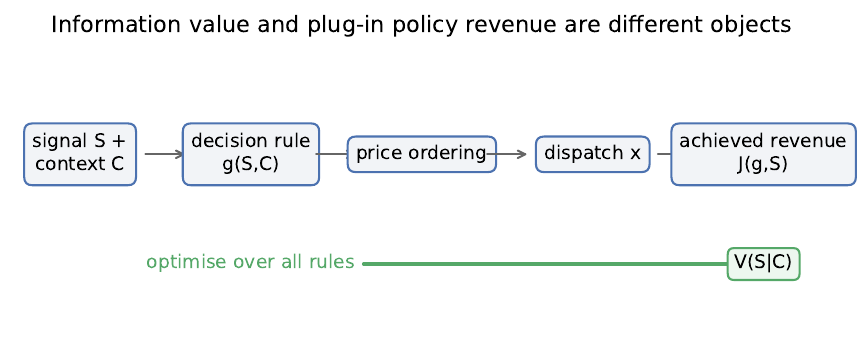}
\caption{Decision chain and the two value objects. The upper path follows one
implemented rule from signal and calendar context to its achieved revenue
$J(g,S)$. Optimising over all admissible rules gives $V(S\mid C)$ instead.
Consequently a plug-in policy can lose revenue even though the optimal value of
the signal cannot fall below climatology.}
\label{fig:decision-chain}
\end{figure*}

\section{An Information Bound}

\subsection{Statement}

\begin{assumption}[Sub-Gaussian prices]\label{as:subgauss}
There exists $s>0$ such that, for every calendar context $c$, every
$1$-Lipschitz function $f$ on $(\mathbb R^n,\|\cdot\|_\infty)$ and every
$\lambda\in\mathbb R$,
$\mathbb E[\exp(\lambda(f(\Pi)-\mathbb E[f(\Pi)\mid C=c]))\mid C=c]
\le\exp(\lambda^2s^2/2)$.
\end{assumption}

Assumption~\ref{as:subgauss} is not restrictive in the sense that it always
holds for bounded prices: if every price lies in an interval of width $\Omega$,
Hoeffding's lemma gives Assumption~\ref{as:subgauss} with $s=\Omega/2$
\cite{hoeffding1963}. It is restrictive in the sense that the useful values of
$s$ are much smaller than $\Omega/2$, and establishing them for a real price law
is an empirical matter that we treat as such.

\begin{theorem}[Euros per nat]\label{th:bound}
Under Assumptions~\ref{as:polytope} and \ref{as:subgauss}, for every signal $S$,
\begin{equation}
  0\ \le\ V(S\mid C)-V_0\ \le\ \min\Bigl\{\,V_\star-V_0,\ \ L\,s\,\sqrt{2\,I(S;\Pi\mid C)}\,\Bigr\},
  \label{eq:bound}
\end{equation}
with $I$ measured in nats.
\end{theorem}

\begin{proof}[Proof sketch]
The left inequality and the first term of the minimum are Jensen. For the second,
$h$ is $L$-Lipschitz for $\|\cdot\|_\infty$, so
$V(S\mid C)-V_0\le L\,\mathbb E\|\mathbb E[\Pi\mid S,C]-\mathbb E[\Pi\mid C]\|_\infty$. Each
coordinate of that difference is the gap between the means of the same
$1$-Lipschitz coordinate map under the posterior and under the prior, hence at
most the Wasserstein-1 distance between them for the $\|\cdot\|_\infty$ metric.
Assumption~\ref{as:subgauss} yields a transport-entropy inequality
\cite{bobkov1999,marton1996}, giving
$W_1(P_{\Pi\mid S=\sigma,C=c},P_{\Pi\mid C=c})\le s\sqrt{2\,\mathrm{KL}}$.
Averaging over $(S,C)$,
using the concavity of the square root, and using
$\mathbb E_{S,C}[\mathrm{KL}(P_{\Pi\mid S,C}\,\|\,P_{\Pi\mid C})]
=I(S;\Pi\mid C)$ completes the
argument. The full proof is in the appendix.
\end{proof}

\begin{corollary}[Information requirement]\label{co:required}
To capture a share $\gamma\in[0,1]$ of the uplift, a signal must carry at least
\begin{equation}
  I \ \ge\ \frac{\bigl(\gamma\,(V_\star-V_0)\bigr)^2}{2L^2s^2}
  \label{eq:required}
\end{equation}
nats. In particular, a claimed revenue improvement implies a lower bound on the
information content of the forecast that produced it, and a forecast that
demonstrably cannot carry that much information cannot have produced that
improvement.
\end{corollary}

\subsection{What the bound does and does not say}

Three readings are legitimate and one is not. The theorem gives a square-root
envelope: no signal with conditional information $I$ can exceed it. It does not
imply decreasing marginal value for an arbitrary family of signals. The
corollary is a falsification tool
that runs in the direction studies rarely test, from a claimed euro figure back
to an implied information requirement. And the bound is computable in advance:
$L$ is a property of the asset, $s$ of the price series, and neither requires
solving an optimisation.

What the bound does not say is that a signal carrying $I$ nats \emph{achieves}
anything. It is an upper limit on what information can buy, not a construction.
The rest of the paper is about the size of that gap, which turns out to be the
interesting quantity.

\subsection{An ordinal counterpart}\label{sec:ordinal}

Theorem~\ref{th:bound} compresses every kind of dependence into one scalar.
The decision instead uses specific distinctions between intervals. This
subsection makes that precise on a
relaxation of the problem, which is what justifies looking at an ordinal
statistic in Section~\ref{sec:results} rather than merely noticing that one
works.

Consider the relaxation $\tilde X\supseteq X$ obtained by keeping the per
interval power limits, the daily discharged-energy budget and the requirement
that charged and discharged energy balance through the efficiency, and dropping
the intra-day path constraint on the state of charge. Write $\eta$ for the
one-way efficiency, so that delivering one unit to the grid consumes
$1/\eta$ units of stored energy and storing one unit consumes $1/\eta$ units
bought.

\begin{proposition}[Ordinal sufficiency in the relaxation]\label{pr:ordinal}
On $\tilde X$, an optimal schedule is obtained by sorting the intervals by
price, pairing the dearest with the cheapest, the second dearest with the second
cheapest, and so on, and executing a pair at full power if and only if its net
margin $\pi_j-\pi_i/\eta^2$ is positive, until the energy budget is exhausted.
Consequently, any
two price vectors that induce the same ordering of the intervals and the same
set of profitable pairs admit a common optimal schedule, whatever their
magnitudes.
\end{proposition}

\begin{proof}[Proof sketch]
Any feasible schedule on $\tilde X$ decomposes into a transport of energy from
charging to discharging intervals. Writing $m_{ij}$ for the grid-side energy
bought in $i$ and sold in $j$, the objective becomes
$\sum_{ij}m_{ij}(\pi_j-\pi_i/\eta^2)$, whose coefficient is \emph{separable} in
$i$ and $j$. The problem therefore splits: at a given total traded mass, sell in
the dearest intervals and buy in the cheapest, and increase the mass while the
marginal pair satisfies $\pi_j>\pi_i/\eta^2$. Both operations read the price
vector only through its ordering and through one ratio test. Appendix~B gives
the argument in full.
\end{proof}

\begin{corollary}[Rank sufficiency at fixed spread]\label{co:ordinal}
If a signal induces the correct ordering and the correct set of profitable
pairs, the bidder attains $V_\star$ on $\tilde X$ regardless of how wrong its
price levels are. An inversion can reverse a scheduled trade; its realised loss
then scales with the traded mass, efficiency and true spread, not with mean
squared error alone.
\end{corollary}

The corollary motivates the empirical section. Cardinal accuracy can coexist
with a damaging ordering, while a badly calibrated signal can preserve the
decisive order. Mutual information distinguishes signals, but its scalar amount
does not identify which price comparisons they preserve. Section~\ref{sec:results}
therefore compares an ordinal diagnostic with an information-content bound when
explaining the revenue of specified plug-in policies.

Proposition~\ref{pr:ordinal} is stated on $\tilde X$ and not on $X$. The
intra-day state of charge path can bind, in which case the greedy pairing is
infeasible and the optimal schedule stops being a pure function of the ordering.
All numerical results use the full $X$, so the proposition is an explanation of
the mechanism and not an assumption of the experiment.

\section{Experimental Design}

The protocol below was written before any result was produced, together with the
falsification clauses of Section~\ref{sec:falsify}.

\subsection{Data and asset}

French day-ahead prices from the Fraunhofer ISE Energy-Charts platform
\cite{energycharts2026} give \NDays{} complete local days between January 2024
and July 2026, hourly until the Single Day-Ahead Coupling transition and
quarter-hourly thereafter. The five clock-change days of that window are dropped
rather than special-cased, because they carry nothing on this question and would
break the per-interval information accounting. The asset is a 1~MW, 2~MWh battery at 85\%
round-trip efficiency, capped at 1.5 equivalent cycles per day, returning to 50\%
state of charge at the start and end of each day.

The prior is the mean price by month and hour, estimated leave-one-out so that
the evaluated settlement price enters neither its mean nor its variance. It explains \ExplainedVar\% of
the variance of the price series, leaving a residual standard deviation of
\ResidualSd~\eur{}/MWh against \PriceSd~\eur{}/MWh for the raw price. Measuring
information against the residual rather than the raw variance is essential: a
signal must not be credited for the diurnal shape that the prior already
contains.

\subsection{Signals and their information content}

The signal is the true price plus independent Gaussian noise,
$S=\Pi+\varepsilon$ with $\varepsilon\sim N(0,\sigma^2 I_n)$, and $\sigma$ is
swept from $1$ to \TailSigma~\eur{}/MWh over eleven levels with three
independent seeds each. This family is chosen for one reason: its information
content admits a closed-form upper bound. The channel is memoryless, so the
mutual information of a day is at most the sum of the per-interval mutual
informations, and each of those is at most the Gaussian channel capacity at that
interval's residual variance, because the Gaussian source maximises it at fixed
variance \cite{cover2006}. Using an upper bound on $I$ on the right-hand side of
\eqref{eq:bound} keeps the inequality valid, so the falsification test remains
sound.

We stress what this family is not. It is a designed experiment, not a model of
real forecast error, which is biased, heteroscedastic and serially correlated.
It buys computability of $I$ at the price of realism, and
Section~\ref{sec:limits} states the cost.

\subsection{Three decision rules}

Equation~\eqref{eq:value} says the posterior mean is sufficient, so the only
question is how a bidder computes it. Three rules span the practice:
\begin{itemize}
\item \emph{naive}: bid the signal, which is what using a forecast as a price
      normally means;
\item \emph{posterior mean, one variance}: linear shrinkage towards the
      climatology with a weight equal to the signal-to-noise ratio computed from
      a single pooled residual variance;
\item \emph{posterior mean, variance by month and hour}: the same with a
      residual variance that depends on month and hour, which is the
      specification the data supports.
\end{itemize}
The last two are Bayes rules under a Gaussian working model. That model is
misspecified, and the comparison between the three is precisely a measurement of
what misspecification costs, since sufficiency of the posterior mean is a
statement about the true posterior and confers nothing on an approximate one.

Every schedule, whatever the rule, is settled at the true prices. In total
\NCases{} backtests of \NDays{} days each are solved.

\section{Results}\label{sec:results}

\subsection{The climatological bidder captures most of the value}

The perfect-foresight bidder earns \StarRevenue{} k\eur{} per megawatt-year of
pure day-ahead arbitrage. The climatological bidder, who knows only the average
price for this month and this hour, earns \BaseRevenue{}, or \BaseShare\% of it.
The uplift available to any forecast of today's prices is therefore
\UpliftRevenue{} k\eur{} per megawatt-year, \UpliftDay~\eur{} per day.

This is not a small correction to the framing; it is the framing. Four fifths of
the arbitrage value of this asset is in the average diurnal shape of the price,
which requires no forecast at all, and the whole of the price-forecasting
literature is competing for the remaining fifth.

\subsection{The bound is numerically uninformative here}

On the empirical distribution, the observed price span
$\Omega=\PriceSpan$~\eur{}/MWh gives the rigorous constant
$s=\Omega/2=\HoeffdingS$~\eur{}/MWh. Corollary~\ref{co:required} then requires
only $\RequiredNatsHoeffding\times10^{-5}$ nats per day to permit the whole
uplift. The pooled residual estimate $s=\EmpiricalS$~\eur{}/MWh would raise that
threshold to $\RequiredNatsEmpirical\times10^{-3}$ nats, but it is not a
joint-law certificate and is reported only as a sensitivity. At the reference
noise level, the Gaussian-capacity upper bound on information is \NatsAtRef{}
nats, and even the ceiling built on that smaller estimated constant is
\SlackAtRef{} times the uplift; the rigorous Hoeffding constant puts it four
times higher still. The
screen therefore cannot discriminate between the tested policies. This is a
consistency check on their achieved revenues, not a direct measurement of the
optimal value $V(S\mid C)$ or a guarantee for future prices.

\subsection{Plug-in policy revenue is not monotone in information}

Fig.~\ref{fig:capture} reports the share of the
uplift each rule captures. Above a noise level of about half the
residual standard deviation, all three rules destroy revenue: they earn less than
the bidder who ignores the signal. At $\sigma=\RefSigma$~\eur{}/MWh, close to
the residual standard deviation of \ResidualSd, the naive rule leaves the asset
at \NaiveAtRef\% of the uplift and the best Bayes rule at \BayeshAtRef\%, both
far below the zero that ignoring the signal would guarantee.

\begin{figure}[t]
\centering
\includegraphics[width=\columnwidth]{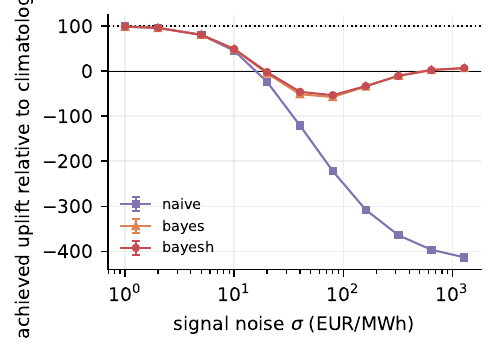}
\caption{Share of the uplift captured against signal noise, three rules, error
bars over three seeds. All three rules cross zero between
$\sigma=\NaiveBreakEven$ and $\sigma=\AllNegativeSigma$~\eur{}/MWh. The two
shrinkage rules reach a minimum at $\sigma=\WorstSigma$ and then recover towards
zero as the shrinkage weight vanishes; the naive rule, which never shrinks,
keeps losing.}
\label{fig:capture}
\end{figure}

Two features of this curve deserve attention. It is not monotone in $\sigma$,
and therefore not monotone in $I$: the worst outcome is in the middle of the
range, not at either end. And the recovery at large $\sigma$ is the
pre-registered convergence check: at $\sigma=\TailSigma$ the shrinkage weight is
near zero and the Bayes rules return to \TailBayesh\% of the uplift, that is, to
the climatological bidder. The naive rule does not recover, because it never
shrinks.

Both Gaussian-shrinkage rules dominate the naive rule at every tested noise
level. This is an implementation diagnostic, not a theorem under the
misspecified Gaussian working model.
Their advantage is large: at $\sigma=\RefSigma$ shrinkage recovers more than
half of the naive rule's loss. Refining the variance from a single pooled value
to one per month and hour adds little, moving \BayesAtRef\% to \BayeshAtRef\%.
The cost of misspecification is therefore concentrated in something other than
the variance profile.

\subsection{The decision is ordinal}

Equation~\eqref{eq:value} tells us where to look. A price-taking storage asset
with a cycle budget does not care how expensive the evening is; it cares which
intervals are the cheapest and which the dearest, with magnitudes entering only
through the threshold that the round-trip efficiency imposes on a profitable
spread. Mutual information aggregates dependence into a scalar: two signals
with identical $I$ can preserve different price comparisons and therefore
induce different schedules.

Fig.~\ref{fig:rank} replaces the horizontal axis by the mean daily rank
agreement between the bid vector and the true price. The picture becomes almost
one-dimensional. Across all three rules and all eleven noise levels, rank
agreement explains \RankRsq{} of the in-sample variation in captured revenue,
while the logarithm of the information upper bound explains \NatsRsq. The climatological bidder sits
at a rank agreement of \ClimatologyRank. Within this designed sweep, every point
above that line gains revenue relative to climatology, and every point below it
loses revenue except the four largest-noise shrinkage points, whose weight on
the signal is small enough that they are indistinguishable from the
climatological bidder in rank agreement and sit within \TailBayesh\% of the
uplift of its revenue. The
separation is therefore between signals that order the day better than the
climatology and signals that order it materially worse; it is descriptive and
need not survive another signal family or sample.

\begin{figure}[t]
\centering
\includegraphics[width=\columnwidth]{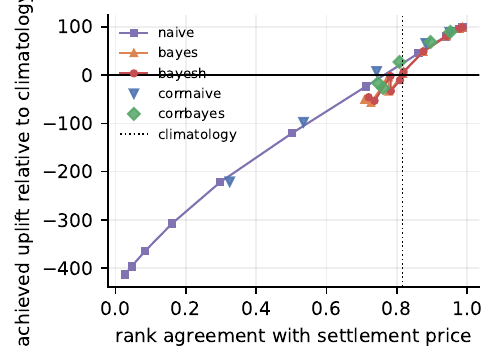}
\caption{Achieved policy revenue against mean daily rank agreement between the
bid vector and settlement price. The dotted line is climatology. The pooled
in-sample association is $R^2=\RankRsq{}$; it is descriptive, not a universal
deployment threshold.}
\label{fig:rank}
\end{figure}

This also explains the non-monotonicity. Adding unbiased Gaussian noise to the
price and shrinking it optimally in the mean squared sense produces an estimate
whose error is small but whose \emph{ordering} is worse than the prior's, because
the prior orders the day correctly on average while a perturbed prior orders a
specific day wrongly with confidence. The revenue is collected through the
ordering, so a small mean squared error bought at the cost of a corrupted
ordering is a bad trade. Mean squared error and revenue disagree not in
magnitude but in what they are functions of.

\subsection{Ordering alone, and one exchange of the extremes}

The previous subsection is a correlation. This one is the experiment that
Corollary~\ref{co:ordinal} predicts, and it is the sharpest evidence in the
paper because it separates ordering from level by construction rather than by
regression.

The \emph{ordinal} bidder receives a vector that has the correct ranking of the
day's intervals and nothing else: the magnitudes are replaced by an arbitrary
linear scale from zero to one hundred, identical on every day, so the vector
carries no cardinal information beyond the ranking: it does not reveal the
price level, spread, or whether trading is profitable. That bidder earns \OrdinalRevenue{}
k\eur{} per megawatt-year, which is \OrdinalCapture\% of the uplift. Knowing
only which hour is dearest and which is cheapest, and nothing about how much,
recovers nine tenths of everything perfect foresight is worth over climatology.

The \emph{swapped} bidder is the mirror image. It receives the exact settlement
prices, correct to the cent in every interval, with one modification: the
cheapest and the dearest entries of the day are exchanged. This swaps two
entries but reverses their pairwise order against every intermediate entry;
every other number is exact. That bidder earns \SwapRevenue{}
k\eur{} per megawatt-year. It does not merely forgo the uplift, which would put
it at zero capture; it captures \SwapCapture\% of the uplift and loses
\SwapLossShare\% of the total revenue, falling far below the bidder that ignores
prices and follows the climatology.

The two experiments bracket the thesis. A signal carrying no level information
at all but the correct ordering is worth almost the whole uplift, while a signal
whose levels are exact to the cent can destroy more than half the revenue by
exchanging two entries. This is a sensitivity test of
the decision rule, not a comparison of the Blackwell value of the signals: a
fully informed bidder that knew the transformation could invert it.

\subsection{Errors that look like real forecast errors}

Independent Gaussian noise is the wrong model of a forecast error, and it was
chosen only because its information content is computable. This subsection drops
that convenience and keeps the ordinal question, which is the trade the paper is
ultimately making.

We repeat the sweep with an error that is serially correlated within the day,
following a first-order autoregression with coefficient $0.8$, and
heteroscedastic, scaled interval by interval by the square root of that hour's
own residual variance. Both features are characteristic of published price
forecasts and neither is representable in the memoryless channel of
Section~V-B. In consequence the mutual information of these signals has no
closed form and we do not report one: their information coordinate is simply
unavailable. That is not a gap in the
experiment; it is the practical asymmetry between the two candidate currencies.
An ordinal statistic can be computed for any signal whatsoever, including a real
forecast, whereas mutual information can be computed only for signals we
ourselves constructed.

The measurements themselves are more forgiving and qualitatively identical. At
the reference noise level the correlated naive bidder captures
\CorrNaiveAtRef\% of the uplift against \NaiveAtRef\% under independent noise,
and the correlated posterior-mean bidder \CorrBayesAtRef\% against
\BayeshAtRef\%. Break-even moves from a root mean squared error of
\NaiveBreakEven{} to \CorrBreakEven~\eur{}/MWh. Realistic error structure
therefore roughly halves the damage without changing its sign: a forecast whose
error is comparable to the climatological residual still leaves this asset worse
off than no forecast at all.

The decisive observation is where these \CorrPoints{} additional points land.
Added to Fig.~\ref{fig:rank}, they fall on the same curve as
the independent-noise points, and the ordinal statistic's explanatory power over
the pooled set is \CorrRankRsq, unchanged from \RankRsq{} on the original set.
A statistic that organises signals of two different noise structures, one of
which has no tractable information content, is doing the work that mutual
information was supposed to do.

\subsection{Robustness to storage duration}

Every result above is for a two-hour asset. Repeating the central cases on a
four-hour asset of the same power changes the levels and not the conclusions.
The climatological bidder earns \LongBase{} k\eur{} per megawatt-year against
\LongStar{} for perfect foresight, so its share is \LongBaseShare\%, against
\BaseShare\% at two hours. The ordinal bidder captures
\LongOrdinalCapture\% of the uplift, against \OrdinalCapture\% at two hours, and
the swapped bidder captures \LongSwapCapture\%. A longer asset trades in more
intervals per day, so exchanging two extreme entries matters proportionally
less. That is the direction the spread term of Corollary~\ref{co:ordinal} points
in, and a useful consistency check on the mechanism, though the corollary bounds
the loss of an inversion rather than deriving its dependence on duration.

\subsection{Why the losses are so large}

The magnitude of the destruction in Fig.~\ref{fig:capture} deserves a
mechanism, because a reader's first instinct is that a bidder acting on a noisy
forecast should earn something between the climatological and the perfect
value, not less than either.

The instinct fails because the loss is not symmetric in the error. A bidder that
misidentifies the cheapest interval does not simply fail to buy cheaply; it
buys expensively, and it then sells the same energy at a price it also
misidentified. The two errors compound within the same trade, and the daily
cycle budget encourages large trades whenever the forecasted spread is positive.
Nothing in the formulation represents confidence, so a wrong ordering can be
executed at the same power as a correct one.

This also explains why shrinkage helps so much and yet not enough. Shrinking the
signal towards the climatology moves the bid vector's ordering back towards the
prior's ordering, which is correct on average, and that recovers more than half
of the naive rule's loss at $\sigma=\RefSigma$. It cannot recover the rest,
because linear shrinkage acts on levels and the damage is in the ordering: a
shrunk vector can preserve an inversion perfectly well. An uncertainty-aware
formulation that reduced traded volume when the ordering is doubtful would
address the residual, and that is a different paper.

\subsection{Falsification clauses}\label{sec:falsify}

Four diagnostics were pre-registered and execute as assertions. Every achieved
gain lies below the theorem's ceiling; no plug-in policy beats perfect
foresight; both Gaussian-shrinkage rules beat the naive rule on the tested grid;
and they converge to climatology as $\sigma$ grows. The first two follow from
the model if implementation and assumptions align. The latter two test the
chosen working rules and are not general consequences of Bayesian sufficiency.

\section{Implications}

\subsection{For forecast evaluation}

A forecast intended for a storage bidder should be evaluated through realised
policy revenue, with the ordering it induces reported alongside price error.
This follows Gneiting's principle
that a point forecast is meaningless without its loss functional
\cite{gneiting2011}, with the uncomfortable corollary that the functional here is
ordinal and no standard score targets it. The immediate practical test is a
comparison, not a universal threshold: in this experiment, no signal added
revenue while ordering the day materially worse than the climatology does,
measured at \ClimatologyRank{} on this market and this sample. This observed
crossing is a diagnostic to validate out of sample, not a sufficient deployment
rule.

\subsection{A deployment threshold in familiar units}

The garbling parameter $\sigma$ is, by construction, the root mean squared error
of the signal, which makes one result directly comparable with the way forecasts
are actually reported. All three rules cross from creating to destroying revenue
between $\sigma=\NaiveBreakEven$ and $\sigma=\AllNegativeSigma$~\eur{}/MWh, against a
climatological root mean squared error of \ResidualSd~\eur{}/MWh on the same
sample.

The threshold is therefore not "beat the climatology" but "beat it by a factor
of two to four", and the gap between those two statements is the whole of the
non-monotonicity described above. A forecast that improves on the climatological
error by twenty percent still leaves this asset worse off than using no forecast
at all. We state the number with the caveat it deserves: it is specific to
Gaussian, unbiased, serially independent errors, and a real forecast with the
same root mean squared error but a better ordering would sit elsewhere. That is
precisely why the threshold should be read on the ordinal axis, and why the next
subsection proposes a statistic for it.

\subsection{A score to use instead}

The recommendation is only useful if it names a computable object, so we name
one. Let $\hat p$ be a forecast for a day and $\pi$ the settlement price. Define
the \emph{spread-weighted rank score}
\begin{equation}
  \mathrm{SR}(\hat p,\pi)=
  \frac{\sum_{i<j}|\pi_j-\pi_i|\,
  \operatorname{sgn}(\hat p_j-\hat p_i)\operatorname{sgn}(\pi_j-\pi_i)}
       {\sum_{i<j}\,|\pi_j-\pi_i|},
  \label{eq:score}
\end{equation}
where the sums run over unordered pairs and ties contribute zero. It equals one
when the forecast orders every pair correctly, minus one when it inverts every
pair, and has expectation zero for an independent random ordering. It differs from a plain
rank correlation in one respect that matters: each pair is weighted by the money
at stake in it, which is aligned with the spread term in
Corollary~\ref{co:ordinal}. It is computable from a forecast
and a settlement series with no optimisation, and it inherits the property that
motivated this paper, namely that it is invariant to any strictly increasing
transformation of the forecast and therefore cannot be improved by calibrating
levels.

We do not claim \eqref{eq:score} is proper in the sense of
\cite{gneiting2007}, and establishing whether a proper score exists for this
functional is open. The score is therefore a candidate diagnostic, not a
validated surrogate for revenue: it should be reported beside a revenue
backtest and compared with the climatological score on the same sample.
Establishing out-of-sample thresholds is left to future work.

\subsection{For valuation studies}

Perfect-foresight backtests are standard and are known to be optimistic. This
paper quantifies the structure of that optimism in a way that
separates two very different components. Four fifths of the arbitrage revenue of
this asset survives the complete removal of information about today. Studies that
report the perfect-foresight number and studies that report a realistic one
differ by at most a fifth of the revenue, not by the order of magnitude the
framing sometimes suggests, and the difference is concentrated entirely in the
ordering of the evening ramp.

\subsection{For the economics of forecasting}

Corollary~\ref{co:required} is a screening tool. A vendor claiming that a
forecast adds a given number of euros per megawatt-year implies, through
\eqref{eq:required}, a minimum information content. On this market the implied
requirement is so small that the test never rejects. In this experiment, the
binding constraint is the ordering preserved by the signal and the plug-in rule
that consumes it, not the theorem's information envelope. That is an unusual conclusion
for an information-theoretic exercise, and it is the reason we report the
negative direction of the result as prominently as the theorem.

\section{Limitations}\label{sec:limits}

The Gaussian garbling is a designed experiment and not a model of forecast
error. Real forecast errors are biased, heteroscedastic, serially correlated,
and correlated with the level of the price. Section~VI-F relaxes the second and
third of these and finds the damage roughly halved and the ordinal collapse
intact, but bias and price-dependence remain untested and we cannot sign their
effect. What survives without qualification is the theorem, which holds for any
signal, and Proposition~\ref{pr:ordinal}, which depends on the structure of the
decision rather than on the noise family.

The information figures are upper bounds, not values. They use the Gaussian
channel capacity at the residual variance, which the empirical source need not attain, so the
reported information coordinates are conservative in the direction that makes
the bound easier to satisfy. The falsification test is therefore
weaker than it would be with exact information, and we state it as a consistency
check rather than as a proof.

Assumption~\ref{as:subgauss} concerns the joint law of the day's prices under the
supremum metric, while the estimate of $s$ we report is computed from the pooled
marginal residual. The two coincide when the daily residual has independent
coordinates and otherwise differ in a direction we cannot sign. This is why every
falsification test uses the unconditional Hoeffding constant, for which
Assumption~\ref{as:subgauss} is a theorem, and the estimated constant appears
only as an indication of how much sharper the bound could be.

Finally, the study covers day-ahead arbitrage for one asset on one market over
thirty-one months, with perfect knowledge of the settlement prices used for
evaluation. Reserve markets, intraday trading, degradation, imbalance settlement
and portfolio effects are all absent, and the climatological prior would be
weaker on a market with less regular diurnal structure than France.

\section{Conclusion}

The optimal value of a price signal to a storage asset lies below a square-root
envelope in its conditional mutual information, with a constant fixed by the
asset's throughput and the price law. On French day-ahead prices this ceiling is
too loose to discriminate between the bidding policies tested here.

What is scarce is ordering. The schedule is decided by which intervals are
cheapest and dearest, distinctions that a scalar amount of mutual information
does not identify and that mean squared error does not protect. We prove ordinal sufficiency on a
relaxation of the problem and then demonstrate it in the sharpest form the data
allows: a signal carrying no cardinal information beyond the ranking
captures \OrdinalCapture\% of the uplift, while a signal that is exact to the
cent except for exchanging the cheapest and dearest entries costs
\SwapLossShare\% of total revenue under the plug-in rule. At the reference noise
level, even Gaussian shrinkage underperforms climatology. Within the designed
sweep, one ordinal statistic has in-sample $R^2=\RankRsq{}$, against
\NatsRsq{} for the information upper bound.

Three questions follow and we leave them open. Whether a proper scoring rule
exists for an ordinal functional of this kind is the first, and it decides
whether \eqref{eq:score} can be more than a diagnostic. Whether the ordinal
sufficiency of Proposition~\ref{pr:ordinal} extends to the full feasible set
once the intra-day state of charge path binds is the second, and it is a
question about how often that constraint is active rather than about whether it
can be. The third is empirical and is the natural continuation: replacing the
designed Gaussian garbling by the errors of published forecasting models
\cite{lago2021} would say where real forecasts sit on
Fig.~\ref{fig:rank}, which is the only place this paper's recommendation can be
tested against practice.

Two practical rules follow. Evaluate forecast-driven schedules against the
climatological policy, and use ordering diagnostics to explain failures rather
than treating them as a universal deployment threshold. When reporting the value of forecasting, report it against a
climatological bidder rather than against zero: on this asset and this market,
four fifths of the arbitrage revenue never depended on knowing anything about
today.

\appendices
\section{Proof of Theorem~\ref{th:bound}}

Fix a context $c$. Write $Q_c=P_{\Pi\mid C=c}$ for the conditional prior and
$P_{\varsigma,c}=P_{\Pi\mid S=\varsigma,C=c}$ for the posterior. By
\eqref{eq:value} and the definition of $V_0$,
\[
  V(S\mid C)-V_0=\mathbb E_{S,C}\bigl[
  h(\mathbb E_{P_{S,C}}\Pi)-h(\mathbb E_{Q_C}\Pi)\bigr].
\]
Since $h$ is the support function of $X$, for any $p,q$,
$|h(p)-h(q)|\le\max_{x\in X}|\langle x,p-q\rangle|
\le\bigl(\max_{x\in X}\|x\|_1\bigr)\|p-q\|_\infty=L\|p-q\|_\infty$, which is
\eqref{eq:throughput}. Hence
\[
  V(S\mid C)-V_0\le L\,\mathbb E_{S,C}\bigl\|
  \mathbb E_{P_{S,C}}\Pi-\mathbb E_{Q_C}\Pi\bigr\|_\infty .
\]
Fix $(\varsigma,c)$ and $t$. The coordinate map $\pi\mapsto\pi_t$ is $1$-Lipschitz
for $\|\cdot\|_\infty$, so by the Kantorovich duality
$\bigl|\mathbb E_{P_{\varsigma,c}}\pi_t-\mathbb E_{Q_c}\pi_t\bigr|
\le W_1(P_{\varsigma,c},Q_c)$, the Wasserstein-1 distance for that metric, and the
bound is uniform in $t$, giving
$\|\mathbb E_{P_{\varsigma,c}}\Pi-\mathbb E_{Q_c}\Pi\|_\infty
\le W_1(P_{\varsigma,c},Q_c)$.

Assumption~\ref{as:subgauss} states that every $Q_c$ satisfies a sub-Gaussian
concentration property for $\|\cdot\|_\infty$-Lipschitz functions with parameter
$s$. By the Bobkov and G\"otze characterisation \cite{bobkov1999}, this is
equivalent to the transport-entropy inequality
\[
  W_1(P,Q)\le s\sqrt{2\,\mathrm{KL}(P\,\|\,Q)}
  \qquad\text{for every }P\ll Q ,
\]
which is Marton's inequality in the form we need \cite{marton1996}. Applying it
at each $(\varsigma,c)$ and averaging,
\begin{align*}
V(S\mid C)-V_0
&\le Ls\,\mathbb E_{S,C}\bigl[
\sqrt{2\,\mathrm{KL}(P_{S,C}\,\|\,Q_C)}\bigr] \\
&\le Ls\sqrt{2\,\mathbb E_{S,C}\bigl[
\mathrm{KL}(P_{S,C}\,\|\,Q_C)\bigr]} .
\end{align*}
the second step by concavity of the square root and Jensen's inequality. The
remaining expectation is the mutual information,
$\mathbb E_{S,C}[\mathrm{KL}(P_{\Pi\mid S,C}\,\|\,P_{\Pi\mid C})]
=I(S;\Pi\mid C)$
\cite[Ch.~2]{cover2006}, which gives the second term of the minimum in
\eqref{eq:bound}.

For the first term, $h$ is convex, so
$h(\mathbb E[\Pi\mid S,C])\le\mathbb E[h(\Pi)\mid S,C]$ and hence
$V(S\mid C)\le V_\star$. Conditional Jensen also gives
$V_0\le V(S\mid C)$. \hfill$\blacksquare$

\section{Proof of Proposition~\ref{pr:ordinal}}

Let $\eta$ be the one-way efficiency, $P$ the power limit, $\Delta t$ the
interval length, and $B$ the daily grid-side discharge budget. On $\tilde X$ a
schedule is described by the charged and discharged energies per interval,
subject to the per-interval caps and to the requirement that the stored energy
balance over the day.

Introduce the transport variables $m_{ij}\ge0$, the grid-side energy bought in
interval $i$ and later sold in interval $j$. Delivering $m$ to the grid at $j$
withdraws $m/\eta$ from storage, and storing that amount required buying
$m/\eta^2$ at $i$, so the revenue of that unit is $\pi_j-\pi_i/\eta^2$ and the
programme is
\begin{align*}
  \max_{m\ge0}\ & \sum_{i,j} m_{ij}\Bigl(\pi_j-\frac{\pi_i}{\eta^2}\Bigr)\\
  \text{s.t.}\ & \sum_{i,j}m_{ij}\le B,\qquad
  \sum_i m_{ij}\le P\Delta t\ \ \forall j,\\
  & \sum_j m_{ij}\le \eta^2 P\Delta t\ \ \forall i .
\end{align*}

The coefficient is separable: $\pi_j-\pi_i/\eta^2=a_j+b_i$ with $a_j=\pi_j$ and
$b_i=-\pi_i/\eta^2$. Writing $D_j=\sum_i m_{ij}$ and $C_i=\sum_j m_{ij}$ for the
discharge and charge masses, the objective equals $\sum_j a_jD_j+\sum_i b_iC_i$
and depends on $m$ only through those two marginals. Any pair of marginals with
equal total mass and satisfying the caps is realisable by some $m\ge0$, for
instance by the north-west corner rule, so the programme is equivalent to
\[
  \max_{m\le B}\ \Bigl[\max_{D:\,\sum D=m}\sum_j\pi_jD_j
  \ -\ \frac1{\eta^2}\min_{C:\,\sum C=m}\sum_i\pi_iC_i\Bigr],
\]
subject to the per-interval caps.

Both inner problems are continuous knapsacks with uniform capacities. Their
solutions fill the intervals in order: $D$ saturates the dearest intervals
first, $C$ the cheapest first. Hence at total mass $m$ the objective is
$\Phi(m)=\int_0^m\bigl(\pi_{(\uparrow)}(u)-\pi_{(\downarrow)}(u)/\eta^2\bigr)
\,du$, where $\pi_{(\uparrow)}$ lists prices in decreasing order and
$\pi_{(\downarrow)}$ in increasing order, both stretched by the interval
capacity. The integrand is non-increasing in $u$, so $\Phi$ is concave and its
maximiser is the largest $m\le B$ at which the integrand is non-negative, that
is, at which the $m$-th dearest price exceeds the $m$-th cheapest divided by
$\eta^2$.

The optimal schedule is therefore determined by the two sorted orders of the
price vector together with the index at which the ratio test
$\pi_{(\uparrow)}>\pi_{(\downarrow)}/\eta^2$ fails. Two price vectors sharing
both objects share an optimal schedule, whatever their magnitudes, which is
the claim. Corollary~\ref{co:ordinal} follows by evaluating the true revenue
$\sum_{ij}m_{ij}(\pi_j-\pi_i/\eta^2)$ at the schedule induced by the signal and
bounding each mismatched transfer by its true coefficient gap. No equality is
claimed when one exchange changes several pairwise orderings.
\hfill$\blacksquare$

\section*{Reproducibility}

The protocol, including the hypothesis and the four falsification clauses, was
written before any experiment was run and is versioned as \texttt{PLAN.md}. The
price series and the dispatch model are public at
\url{https://github.com/brieuclerouxtardif-blip/bess-arbitrage-fr}, at commit
\texttt{7742ad9}; the experiment definitions, the raw case-by-case results, the
seeds and the figure scripts are available from the author.
Every numerical value in the manuscript
is read from a single generated macro file, so no number is typed by hand and a
change of calibration cannot leave a stale figure behind. Fig.~1 is a schematic
and is the one drawing not produced by that pipeline.

\bibliographystyle{IEEEtran}
\IEEEtriggeratref{10}
\bibliography{refs}

\end{document}